\documentclass[runningheads, a4paper]{llncs}

\usepackage{graphicx}
\usepackage{microtype}
\usepackage{multirow}
\usepackage{booktabs}
\usepackage[utf8]{inputenc} 
\usepackage{hyperref}
\usepackage{amsmath}
\usepackage{subcaption}
\usepackage{algorithm}
\usepackage[algo2e]{algorithm2e} 
\SetKwRepeat{Do}{repeat}{until}
\usepackage{todonotes}
\usepackage{tikz}
\usetikzlibrary{automata,positioning}
\usepackage{amssymb, amsmath}
\newcommand{\nat}{{\mathbb N}}

\begin{document}

\title{Learning Unions of $k$-Testable Languages \thanks{This research is supported by the Dutch Technology Foundation (STW) under the Robust CPS program (project 12693).}}
\author{Alexis Linard\inst{1} \and
Colin de la Higuera \inst{2} \and
Frits Vaandrager\inst{1}}
\authorrunning{A. Linard, C. de la Higuera, and F. Vaandrager}
\institute{Institute for Computing and Information Science \\ Radboud University, Nijmegen, The Netherlands \\
\email{$\{$a.linard,f.vaandrager$\}$@cs.ru.nl}
\and
Laboratoire des Sciences du Numérique de Nantes \\ Université de Nantes, France \\
\email{cdlh@univ-nantes.fr}}

\maketitle
\setcounter{footnote}{0}

\begin{abstract}
A classical problem in grammatical inference is to identify a language from a set of examples. 
In this paper, we address the problem of identifying a union of languages from examples that belong to several \emph{different} unknown languages.
Indeed, decomposing a language into smaller pieces that are easier to represent should make learning easier than aiming for a too generalized language.
In particular, we consider $k$-testable languages in the strict sense ($k$-TSS).
These are defined by a set of allowed prefixes, infixes (sub-strings) and suffixes that words in the language may contain.
We establish a Galois connection between the lattice of all languages over alphabet
$\Sigma$, and the lattice of $k$-TSS languages over $\Sigma$. We also define a
simple metric on $k$-TSS languages. The Galois connection and the metric allow us to derive an efficient algorithm to learn the union of $k$-TSS languages. 
We evaluate our algorithm on an industrial dataset and thus demonstrate the relevance of our approach. 
\keywords{grammatical inference \and $k$-testable languages \and union of languages \and Galois connection}%mandatory
\end{abstract}

\bibliographystyle{splncs04}

%=====================================
%			INTRO
%=====================================
\section{Introduction}
\label{sec:intro}
A common problem in grammatical inference is to find, i.e.\ \emph{learn}, a regular language from a set of examples of that language.
When this set is divided into positive examples (belonging to the language) and negative examples (not belonging to the language), the problem is typically solved by searching for the smallest deterministic finite automaton (DFA) that accepts the positive examples, and rejects the negative ones.
Moreover there exist algorithms which \emph{identify in the limit} a DFA, that is, they eventually learn correctly any language/automaton from such examples \cite{Gold1967}. 

We consider in this work a setting where one can observe positive examples from multiple different languages, but they are given together and it is not clear to which language each example belongs to.
For example, given the following set of strings $S = \{aa,aaa,aaaa,abab,ababab,$ $abba,abbba,abbbba\}$, learning a single automaton will be less informative than learning several DFAs  encoding respectively the languages $a^*$, $(ab)^*$ and $ab^*a$.
There is a trade-off between the number of languages and how specific each language should be. That is, covering all words through a single language may not be the desired result, but having a language for each word may also not be desired.
The problem at hand is therefore double: to cluster the examples and learn the corresponding languages.

In this paper, we focus on $k$-testable languages in the strict sense ($k$-TSS) \cite{McNaughtonPapert}.
A $k$-TSS language is determined by a finite set of substrings of length at most $k$ that are allowed to appear in the strings of the language. 
It has been proved that, unlike for regular languages, algorithms can learn $k$-TSS languages in the limit from text \cite{YokomoriKobayashi}.
Practically, this learning guarantee has been used in a wide range of applications \cite{bex2006inference,coste2016learning,rogers2011aural,tantini2010sequences}.
However, all these applications consider learning of a sole  $k$-TSS language \cite{bex2006inference}, or the training of several  $k$-TSS languages in a context of supervised learning  \cite{tantini2010sequences}.
Learning unions of $k$-TSS languages has been suggested in \cite{torres2001k}.

A first contribution of this paper is a Galois connection between the lattice of all languages over alphabet $\Sigma$ and the lattice of $k$-TSS languages over $\Sigma$.
This result provides a unifying and abstract perspective on known properties of $k$-TSS languages, but also leads to several new insights.
The Galois connection allows to give an alternative proof of the learnability in the limit of $k$-TSS languages, and suggests an algorithm for learning unions of  $k$-TSS languages.
A second contribution is the definition of a simple metric on $k$-TSS languages.
Based on this metric, we define a clustering algorithm that allows us
to efficiently learn unions of $k$-TSS languages.

Our research was initially motivated by a case study of print jobs that are submitted to large industrial printers.
These print jobs can be represented by strings of symbols, where each symbol denotes a different media type, such as a book cover or a newspaper page.
Together, this set of print jobs makes for a fairly complicated `language'.
Nevertheless, we observed that each print job can be classified as belonging to one of a fixed set of categories, such as `book' or `newspaper'.
Two print jobs that belong to the same category are typically similar, to the extent that they only differ in terms of prefixes, infixes and suffixes.
Therefore, the languages stand for the different families of print jobs.
Our goal is to uncover these $k$-TSS languages.

This paper is organized as follows. In Section \ref{sec:def} we recall preliminary definitions on $k$-TSS languages and define a Galois connection that characterizes these languages.
We then present in Section \ref{sec:method} our algorithm for learning unions of $k$-TSS languages.
Finally, we report on the results we achieved for the industrial case study in Section \ref{sec:experiments}.

%=====================================
%			DEFINITIONS
%=====================================
\section{$k$-Testable Languages}
\label{sec:def}
The class of $k$-testable languages in the strict sense ($k$-TSS) has been introduced by McNaughton and Papert \cite{McNaughtonPapert}.
Informally, a $k$-TSS language is determined by a finite set of substrings of length at most $k$ that are allowed to appear in the strings of the language. 
This makes it possible to use as a parser a sliding window of size $k$, which rejects the strings that at some point do not comply with the conditions.
Concepts related to  $k$-TSS languages have been widely used e.g.\ in information theory, pattern recognition and DNA sequence analysis
\cite{GarciaVidal,YokomoriKobayashi}.
Several definitions of $k$-TSS languages occur in the literature, but the differences are technical.
In this section, we present a slight variation of the definition of $k$-TSS languages from \cite{cdlh}, 
which in turn is a variation of the definition occurring in \cite{GarciaVidal,GarciaVidalOncina}.
We establish a Galois connection that characterizes $k$-TSS languages, and show how this Galois connection may
be used to infer a learning algorithm.

We write $\nat$ to denote the set of natural numbers, and let $i$, $j$, $k$, $m$, and $n$ range over $\nat$.
\subsection{Strings}
Throughout this paper, we fix a finite set $\Sigma$ of \emph{symbols}. 
A \emph{string} $x = a_1 \ldots a_n$ is a finite sequence of symbols.
The \emph{length} of a string $x$, denoted $\mid x \mid$ is the number of symbols occurring in it.
The empty string is denoted $\lambda$.
We denote by $\Sigma^{\ast}$ the set of all strings over $\Sigma$, and by $\Sigma^{+}$ the set of all nonempty strings over $\Sigma$ (i.e.\ $\Sigma^{\ast} = \Sigma^{+} \cup \{ \lambda \}$).
Similarly, we denote by $\Sigma^{<i}$, $\Sigma^i$ and $\Sigma^{>i}$ the sets of strings over $\Sigma$ of length less than $i$, equal to $i$, and greater than $i$, respectively.

Given two strings $u$ and $v$, we will denote by $u \cdot v$ the concatenation of $u$ and $v$.
When the context allows it, $u \cdot v$ shall be simply written $u v$.
We say that $u$ is a \emph{prefix} of $v$ iff there exists a string $w$ such that $u w = v$.
Similarly, $u$ is a \emph{suffix} of $v$ iff there exists a string $w$ such that $w u  = v$.
We denote by $x[:k]$ the prefix of length $k$ of $x$ and $x[-k:]$ the suffix of length $k$ of $x$.

A \emph{language} is any set of strings, so therefore a subset of $\Sigma^{\ast}$.
Concatenation is lifted to languages by defining $L \cdot L' = \{ u \cdot v \mid u \in L \mbox{ and } v \in L' \}$.
Again, we will write $L L'$ instead of $L \cdot L'$ when the context allows it.

\subsection{$k$-Testable Languages}

\label{subsec:k-testable-languages}
A $k$-TSS language is determined by finite sets of strings of length $k-1$ or $k$ that are allowed as prefixes, suffixes and substrings, respectively, together with all the short strings (with length at most $k-1$) contained in the language.
The finite sets of allowed strings are listed in what McNaughton and Papert \cite{McNaughtonPapert} called a \emph{$k$-test vector}.
The following definition is taken from \cite{cdlh}, except that we have omitted the fixed alphabet $\Sigma$ as an element in the tuple,
and added a technical condition ($I \cap F = C \cap \Sigma^{k-1}$) that we need to prove Theorem~\ref{Galois}.

\begin{definition}
Let $k >0$. A \emph{$k$-test vector} is a 4-tuple $Z = \langle I, F, T, C \rangle$ where
\begin{itemize}
\item 
$I \subseteq \Sigma^{k-1}$ is a set of allowed \emph{prefixes},
\item
$F \subseteq \Sigma^{k-1}$ is a set of allowed \emph{suffixes},
\item
$T \subseteq \Sigma^k$ is a set of allowed \emph{segments}, and
\item
$C \subseteq \Sigma^{<k}$ is a set of allowed \emph{short strings} satisfying $I \cap F = C \cap \Sigma^{k-1}$.
\end{itemize}
We write ${\cal T}_k$ for the set of $k$-test vectors. 
\end{definition}

Note that the set ${\cal T}_k$ of $k$-test vectors is finite. We equip set ${\cal T}_k$ with
a partial order structure as follows.

\begin{definition}
Let $k >0$. The relation $\sqsubseteq$ on ${\cal T}_k$ is given by
\begin{eqnarray*}
\langle I, F, T, C \rangle \sqsubseteq \langle I', F', T', C' \rangle & \Leftrightarrow & I \subseteq I' \mbox{ and } F \subseteq F' \mbox{ and } T \subseteq T' \mbox{ and } C \subseteq C'.
\end{eqnarray*}
With respect to this ordering, ${\cal T}_k$ has a least element $\bot = \langle \emptyset, \emptyset, \emptyset, \emptyset \rangle$ and a greatest element $\top = \langle \Sigma^{k-1}, \Sigma^{k-1}, \Sigma^k, \Sigma^{<k} \rangle$.
The \emph{union}, \emph{intersection} and \emph{symmetric difference} of two $k$-test vectors $Z = \langle I, F, T, C\rangle$ and $Z' = \langle I', F', T', C'\rangle$ are given by, respectively,
\begin{eqnarray*}
Z \sqcup Z' & = &  \langle I \cup I', F \cup F', T \cup T', C \cup C' \cup (I \cap F') \cup (I' \cap F) \rangle,\\
Z \sqcap Z' & = &  \langle I \cap I', F \cap F', T \cap T', C \cap C' \rangle,\\
Z \bigtriangleup Z' & = &  \langle I \bigtriangleup I', F \bigtriangleup F', T \bigtriangleup T', C \bigtriangleup C' \bigtriangleup (I' \cap F) \bigtriangleup (I \cap F') \rangle.
\end{eqnarray*}
\end{definition}

The reader may check that $Z \sqcup Z'$, $Z \sqcap Z'$ and
$Z \bigtriangleup Z'$ are $k$-test vectors indeed, preserving the property
$I \cap F = C \cap \Sigma^{k-1}$. The reader may also check that
$( {\cal T}_k, \sqsubseteq)$ is a lattice with $Z \sqcup Z'$ the
least upper bound of $Z$ and $Z'$, and $Z \sqcap Z'$ the greatest lower bound of $Z$ and $Z'$.
The symmetric difference operation $\bigtriangleup$ will be used further 
on to define a metric on $k$-test vectors.

We can associate a $k$-test vector $\alpha_k(L)$ to each language $L$ by taking
all prefixes of length $k-1$ of the strings in $L$,
all suffixes of length $k-1$ of the strings in $L$ , and
all substrings of length $k$ of the strings in $L$. 
Any string which is both an allowed prefix and an allowed suffix is also a short string, 
as well as any string in $L$ with length less than $k-1$.
\begin{definition}
Let $L \subseteq \Sigma^{\ast}$ be a language and $k \in \nat$.
Then $\alpha_k (L)$ is the $k$-test vector  $\langle I_k(L), F_k(L), T_k(L), C_k(L) \rangle$ where
\begin{itemize}
\item 
$I_k(L) = \{ u \in \Sigma^{k-1} \mid \exists v \in \Sigma^{\ast} : u v \in L \}$,
\item 
$F_k(L) = \{ w \in \Sigma^{k-1} \mid \exists v \in \Sigma^{\ast} : v w \in L \}$,
\item 
$T_k(L) = \{ v \in \Sigma^k \mid \exists u, w \in \Sigma^{\ast} : u v w \in L \}$, and
\item
$C_k(L) = (L \cap \Sigma^{<k-1}) \cup (I_k(L) \cap F_k(L))$.
\end{itemize}
\end{definition}

It is easy to see that operation $\alpha_k : 2^{\Sigma^{\ast}} \rightarrow {\cal T}_k$ is monotone.
\begin{proposition}
\label{alpha monotone}
For all languages $L, L'$ and for all $k>0$,
\begin{eqnarray*}
L \subseteq L' & \Rightarrow & \alpha_k(L) \sqsubseteq \alpha_k(L').
\end{eqnarray*}
\end{proposition}

Conversely, we associate a language $\gamma_k(Z)$ to each $k$-test vector $Z= \langle I, F, T, C \rangle$, consisting
of all the short strings from
 $C$ together with all strings of length at least $k-1$ whose prefix of length $k-1$ is in $I$, whose suffix of length $k-1$ is in $F$, and
where all substrings of length $k$ belong to $T$.
\begin{definition}
Let $Z = \langle I, F, T, C \rangle$ be a $k$-test vector, for some $k>0$.
Then
\begin{eqnarray*}
\gamma_k(Z) & = & C \cup ((I \Sigma^{\ast} \cap \Sigma^{\ast} F) \setminus (\Sigma^{\ast} (\Sigma^k \setminus T)\Sigma^{\ast})).
\end{eqnarray*}
We say that a language $L$ is \emph{$k$-testable in the strict sense ($k$-TSS)} if there exists a $k$-test vector $Z$ such that $L = \gamma_k (Z)$. Note that all
$k$-TSS languages are regular.
\end{definition}

Again, it is easy to see that operation $\gamma_k : {\cal T}_k \rightarrow 2^{\Sigma^{\ast}}$ is monotone.

\begin{proposition}
\label{gamma monotone}
For all $k>0$ and for all $k$-test vectors $Z$ and $Z'$,
\begin{eqnarray*}
Z \sqsubseteq Z' & \Rightarrow & \gamma_k(Z) \subseteq \gamma_k (Z').
\end{eqnarray*}
\end{proposition}

The next theorem, which is our main result about $k$-testable languages, asserts that $\alpha_k$ and $\gamma_k$ form a (monotone) Galois connection \cite{NNH99} between lattices $({\cal T}_k, \sqsubseteq)$ and $(2^{\Sigma^{\ast}}, \subseteq)$.

\begin{theorem}[Galois connection]
\label{Galois}
Let $k>0$, let $L \subseteq \Sigma^{\ast}$ be a language, and let $Z$ be a $k$-test vector. Then
$\alpha_k (L) \sqsubseteq Z ~ \Leftrightarrow ~ L \subseteq \gamma_k (Z)$.
\end{theorem}
\begin{proof}
Let $Z = \langle I, T, F, C \rangle$.

\paragraph{$\Rightarrow$.}
Assume $\alpha_k (L) \sqsubseteq Z$ and $w \in L$. We need to show that $w \in \gamma_k (Z)$. Since $\alpha_k (L) \sqsubseteq Z$ we know
that $I_k(L) \subseteq I$, $F_k(L) \subseteq F$, $T_k(L) \subseteq T$ and $C_k(L) \subseteq C$.
If $\mid w \mid < k-1$ then $w \in C_k(L) \subseteq C \subseteq \gamma_k (Z)$, and we are done. 
If $\mid w \mid = k-1$ then $w \in I_k(L)$ and $w \in F_k(L)$. This implies that $w \in C_k(L)$.
Now we use again that $C_k(L) \subseteq C \subseteq \gamma_k (Z)$, and we are done.
If $\mid w \mid \geq k$,
then $I_k(L) \subseteq I$ implies that the prefix of length $k-1$ of $w$ is in $I$,
$F_k(L) \subseteq F$ implies that the suffix of length $k-1$ of $w$ is in $F$, and
$T_k(L) \subseteq T$ implies that any substring of length $k$ of $w$ is in $T$.
Thus $w \in (I \Sigma^{\ast} \cap \Sigma^{\ast} F) \setminus (\Sigma^{\ast} (\Sigma^k \setminus T)\Sigma^{\ast})$ and
this implies $w \in \gamma_k (Z)$, as required.

\paragraph{$\Leftarrow$.}
Assume $L \subseteq \gamma_k (Z)$. We need to show that $\alpha_k (L) \sqsubseteq Z$.
This is equivalent to showing
that $I_k(L) \subseteq I$, $F_k(L) \subseteq F$, $T_k(L) \subseteq T$ and $C_k(L) \subseteq C$.
\begin{itemize}
\item 
$I_k(L) \subseteq I$. Suppose $u \in I_k(L)$. Then $\mid u \mid = k-1$ and there exists a string $v$ such that $uv \in L$.
By the assumption, $uv \in \gamma_k (Z)$. We consider two cases: $v = \lambda$ and $v \neq \lambda$.
If $v = \lambda$ then $u \in \gamma_k (Z)$. Since $\mid u \mid = k-1$, this means that $u \in C \cup (I \cap F)$.
Since $Z$ is a $k$-test vector, $I \cap F = C \cap \Sigma^{k-1}$. We can now infer $u \in I \cap F$.
Thus, in particular, $u \in I$, as required.
If $v \neq \lambda$ then, using $uv \in \gamma_k (Z)$, we conclude that $uv \in I \Sigma^{\ast}$.
This implies $u \in I$, as required.
\item
$F_k(L) \subseteq F$. Suppose $w \in F_k(L)$. Then $\mid w \mid = k-1$ and there exists a string $v$ such that $vw \in L$.
By the assumption, $vw \in \gamma_k (Z)$. We consider two cases: $v = \lambda$ and $v \neq \lambda$.
If $v = \lambda$ then $w \in \gamma_k (Z)$. Since $\mid w \mid = k-1$, this means that $w \in C \cup (I \cap F)$.
Since $I \cap F = C \cap \Sigma^{k-1}$, we conclude $w \in I \cap F$.
Thus, in particular, $w \in F$, as required.
If $v \neq \lambda$ then, using $vw \in \gamma_k (Z)$, we conclude that $vw \in \Sigma^{\ast} F$.
This implies $w \in F$, as required.
\item
$T_k(L) \subseteq T$. Suppose $v \in T_k(L)$. Then $\mid v \mid = k$ and there exist strings $u$ and $w$ such that $u v w \in L$.
By the assumption, $u v w \in \gamma_k (Z)$.
As the length of $u v w$ is at least $k$, 
$u v w \in (I \Sigma^{\ast} \cap \Sigma^{\ast} F) \setminus (\Sigma^{\ast} (\Sigma^k \setminus T)\Sigma^{\ast})$.
Thus $u v w \not\in (\Sigma^{\ast} (\Sigma^k \setminus T)\Sigma^{\ast})$.
Therefore, $u v w$ does not contain a substring of length $k$ that is not in $T$.
But this implies $v \in T$, as required.
\item
$C_k(L) \subseteq C$. Suppose $w \in C_k(L)$. Then $w \in L \cap \Sigma^{<k-1}$ or $w \in I_k(L) \cap F_k(L)$.
If $w \in L \cap \Sigma^{<k-1}$ then by the assumption $w \in \gamma_k (Z)$, and, as the length of $w$ is less than $k-1$, we conclude $w \in C$, as required.
If $w \in I_k(L) \cap F_k(L)$ then $w \in I_k(L)$ and $w \in F_k(L)$. By the previous items it then follows that $w \in I$ and $w \in F$,
and therefore $w \in I \cap F$. Now we use $I \cap F = C \cap \Sigma^{k-1}$ to conclude $w \in C$.
\end{itemize}
\end{proof}

The above theorem generalizes results on strictly $k$-testable languages
from \cite{GarciaVidal,YokomoriKobayashi}. 
Composition $\gamma_k \circ \alpha_k$ is commonly called the \emph{associated closure operator}, and
composition $\alpha_k \circ \gamma_k$ is known as the \emph{associated kernel operator}. The fact that we have a Galois connection
has some well-known consequences for these associated operators.

\begin{corollary}
For all $k>0$, $\gamma_k \circ \alpha_k$ and $\alpha_k \circ \gamma_k$ are monotone and idempotent.
\end{corollary}
Monotony of $\gamma_k \circ \alpha_k$ was established previously as Theorem 3.2 in \cite{GarciaVidal} and as Lemma 3.3 in \cite{YokomoriKobayashi}.

\begin{corollary} 
\label{flationary}
For all $k>0$, $L \subseteq \Sigma^{\ast}$ and $Z \in {\cal T}_k$,
\begin{eqnarray}
\label{prop1}
\alpha_k \circ \gamma_k (Z) & \sqsubseteq & Z\\
\label{prop2}
L & \subseteq & \gamma_k \circ \alpha_k(L)
\end{eqnarray}
\end{corollary}
Inequality (\ref{prop1}) asserts that the associated kernel operator $\alpha_k \circ \gamma_k$ is \emph{deflationary}, while
inequality (\ref{prop2}) says that the associated closure operator $\gamma_k \circ \alpha_k$ is \emph{inflationary} (or \emph{extensive}). 
Inequality (\ref{prop2}) was established previously as Lemma 3.1 in \cite{GarciaVidal} and (also) as Lemma 3.1 in \cite{YokomoriKobayashi}.

Another immediate corollary of the Galois connection is that in fact $\gamma_k \circ \alpha_k(L)$ is the smallest $k$-TSS language that contains $L$.
This has been established previously as Theorem 3.1 in \cite{GarciaVidal}.

\begin{corollary}
\label{smallestlanguage}
For all $k>0$, $L \subseteq \Sigma^{\ast}$, and $Z \in {\cal T}_k$,
\begin{eqnarray*}
L \subseteq \gamma_k(Z) & \Rightarrow & \gamma_k \circ \alpha_k (L) \subseteq \gamma_k(Z).
\end{eqnarray*}
\end{corollary}

As a final corollary, we mention that $\alpha_k \circ \gamma_k (Z)$ is the smallest $k$-test vector that denotes the same language as $Z$.
This is essentially Lemma 1 of \cite{YokomoriKobayashi}.

\begin{corollary}
For all $k>0$ and $Z \in {\cal T}_k$,
$\gamma_k \circ \alpha_k \circ \gamma_k (Z) = \gamma_k(Z)$. Moreover, for any $Z' \in {\cal T}_k$,
\begin{eqnarray*}
\gamma_k(Z) = \gamma_k(Z') & \Rightarrow & \alpha_k \circ \gamma_k (Z) \sqsubseteq Z'.
\end{eqnarray*}
\end{corollary}

We can provide a simple characterization of $\alpha_k \circ \gamma_k (Z)$
as the $k$-test vector obtained by removing all the allowed prefixes, suffixes
and segments that do not occur in the $k$-testable language generated by $Z$.

\begin{definition}
Let $Z = \langle I, F, T, C \rangle$ be a $k$-test vector, for some $k>0$.
We say that $u \in I$ is a \emph{junk prefix} of $Z$ if $u$ does not occur as a
prefix of any string in $\gamma_k(Z)$.
Similarly, we say that $u \in F$ is a \emph{junk suffix} of $Z$ if $u$ does not occur as a
suffix of any string in $\gamma_k(Z)$, and we say that $u \in T$ is a \emph{junk segment} of $Z$ if $u$ does not occur as a
substring of any string in $\gamma_k(Z)$.
We call $Z$ \emph{canonical} if it does not contain any junk prefixes,
junk suffixes, or junk segments.
\end{definition}

\begin{proposition}
Let $Z$ be a $k$-test vector, for some $k>0$, and let $Z'$ be the canonical $k$-test vector obtained from $Z$ by deleting all junk prefixes, junk suffixes, and junk segments. Then $\alpha_k \circ \gamma_k (Z) = Z'$.
\end{proposition}

\subsection{Learning $k$-TSS Languages}

It is well-known that any $k$-TSS language can be identified in the limit
from positive examples \cite{GarciaVidal,GarciaVidalOncina}. Below we recall the basic
argument; we refer to \cite{GarciaVidal,GarciaVidalOncina,YokomoriKobayashi} for
efficient algorithms.

\begin{theorem}
\label{kTSSs learnable in the limit}
Any $k$-TSS language can be identified in the limit from positive examples.
\end{theorem}
\begin{proof}
Let $L$ be a $k$-TSS language and let $w_1, w_2, w_3,\ldots$
be an enumeration of $L$. Let $L_0 = \emptyset$ and $L_i = L_{i-1} \cup \{ w_i \}$, for $i>0$. We then have
\[
L_1 \subseteq L_2 \subseteq L_3 \subseteq \cdots
\]
By monotonicity of $\alpha_k$ (Proposition~\ref{alpha monotone}) we obtain
\begin{equation}
\label{chain of k-test vectors}
\alpha_k (L_1) \sqsubseteq \alpha_k (L_2) \sqsubseteq \alpha_k (L_3) \sqsubseteq \cdots
\end{equation}
and by monotonicity of $\gamma_k$ (Proposition~\ref{gamma monotone})
\begin{equation}
\label{chain of k-TSSs}
\gamma_k \circ \alpha_k (L_1) \subseteq \gamma_k \circ \alpha_k (L_2) \subseteq \gamma_k \circ \alpha_k (L_3) \subseteq \cdots
\end{equation}
Since $\gamma_k \circ \alpha_k$ is inflationary (Corollary~\ref{flationary}),
$L$ is a $k$-TSS language and, for each $i$, $\gamma_k \circ \alpha_k (L_i)$ is the
smallest $k$-TSS language that contains $L_i$ (Corollary~\ref{smallestlanguage}), we have
\begin{equation}
\label{sandwich}
L_i \subseteq \gamma_k \circ \alpha_k (L_i) \subseteq L
\end{equation}
Because $({\cal T}_k, \sqsubseteq)$ is a finite partial order it does not have
an infinite ascending chain. This means that sequence (\ref{chain of k-test vectors}) converges. But then sequence (\ref{chain of k-TSSs}) also converges, that is, there exists an $n$ such that, for all $m \geq n$,
$\gamma_k \circ \alpha_k (L_m) = \gamma_k \circ \alpha_k (L_n)$.
By equations (\ref{chain of k-TSSs}) and (\ref{sandwich}) we obtain, for all $i$,
\[
L_i \subseteq  \gamma_k \circ \alpha_k (L_i) \subseteq \gamma_k \circ \alpha_k (L_n) \subseteq L
\]
This implies $L = \gamma_k \circ \alpha_k (L_n)$, meaning that the sequence   (\ref{chain of k-TSSs}) of $k$-TSS languages converges to $L$.
\end{proof}

\section{Learning Unions of $k$-TSS Languages}
\label{sec:method}

In this section, we present guarantees concerning learnability in the limit of unions of $k$-TSS languages.
Then, we present an algorithm merging closest and compatible $k$-TSS languages.

\subsection{Generalities}

It is well-known that the class of $k$-testable languages in the strict sense
is not closed under union. 
Take for instance the two $3$-testable languages, represented by their DFA's in Figure \ref{test:Bild1}, that are generated by the following $3$-test vectors:
\begin{eqnarray*}
Z & = & \langle \{a a\}, \{a a\}, \{a a a\},  \{a a\}\rangle\\
Z' & = & \langle \{b a, b b\}, \{a b, b b \}, \{b a a, b a b, a a a, a a b\}, \{ b b \} \rangle 
\end{eqnarray*}
with $\Sigma=\{a,b\}$. 
The union $\gamma_3(Z) \cup\gamma_3(Z')$ of these languages,
represented by its DFA in Figure \ref{test:Bild2}, is not a $3$-testable language. Indeed, it is not a $k$-testable language for any value of $k>0$. 
For $k=1$, the only $k$-testable language that extends 
$\gamma_3(Z) \cup\gamma_3(Z')$ is $\Sigma^{\ast}$.
For $k\geq 2$, the problem is that since $a^{k-1}$ is an allowed prefix,
$a^{k-1} b$ is an allowed segment, and $a^{k-2} b$ is an allowed suffix,
$a^{k-1} b$ has to be in the language, even though it is not an element of
$\gamma_3(Z) \cup\gamma_3(Z')$.

It turns out that we can generalize Theorem~\ref{kTSSs learnable in the limit} to
unions of $k$-TSS languages.

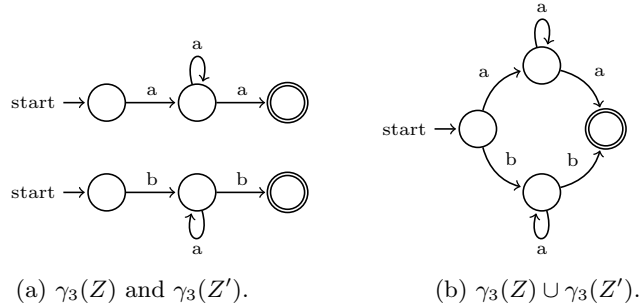
\begin{figure}[h!]
\begin{scriptsize}
\centering     %%% not \center
\captionsetup[subfigure]{position=b}
\subcaptionbox{$\gamma_3(Z)$ and $\gamma_3(Z')$.\label{test:Bild1}}{
\hspace*{1cm}
\begin{tikzpicture}[->,shorten >=1pt,auto,node distance=1.7cm,semithick]
      \node[initial,state,scale=0.7] (init)                    {};
      \node[state,scale=0.7]          (A) [right of=init] {};
      \node[state,accepting,scale=0.7]          (AA) [right of=A] {};
      	 \node[initial,state,scale=0.7] (init2)  [below of=init]                  {};
      \node[state,scale=0.7]          (B2) [right of=init2] {};
       \node[state,accepting,scale=0.7]         (BA2) [right of=B2] {};
      \path (init) edge              node {a} (A)
            (A)    edge				 node {a} (AA)
            (A)   edge [loop above] node {a} (A)
             (A)   edge [loop below,draw=none] node {} (A)
      		(init2) edge [] node {b} (B2)
            (B2)    edge	[loop below] node {a} (B2)   
            (B2)    edge	[loop above,draw=none] node {}(B2)   
            (B2)    edge	[]	node {b} (BA2)
            ;
    \end{tikzpicture}
    \hspace*{0.5cm}
}
\subcaptionbox{$\gamma_3(Z) \cup\gamma_3(Z')$.\label{test:Bild2}}{
    \begin{tikzpicture}[->,shorten >=1pt,auto,node distance=1.7cm,semithick]
      \node[initial,state,scale=0.7] (init)                    {};
      \node[state,scale=0.7]          			(A) [below right of=init] {};
      \node[state,scale=0.7]          (B) [above right of=init] {};
      \node[state,accepting,scale=0.7]          (AA) [above right of=A] {};
      \path (init) edge  [bend left]            node {a} (B)
            (B)    edge	[loop above] node {a} (B)
            (B)    edge	[bend left] node {a} (AA)
            (init) edge  [bend right]            node {b} (A)
              (A) edge  [bend right]            node {b} (AA)
              (A) edge  [loop below]            node {a} (A)
            ;
    \end{tikzpicture}
    \hspace*{1cm}
}
\caption{$k$-testable languages are not closed under union.}
\end{scriptsize}
\label{fig:notclosedunderunion}
\end{figure}

\begin{theorem}
\label{unions of kTSSs learnable in the limit}
Any language that is a union of $k$-TSS languages can
be identified in the limit from positive examples.
\end{theorem}
\begin{proof}
Let $L = L_1 \cup \cdots \cup L_l$, where all the $L_p$ are $k$-TSS
languages, and let $w_1, w_2, w_3,\ldots$ be an enumeration of $L$.
Define, for $i>0$,
\begin{equation*}
    K_i = \bigcup_{j=1}^i \gamma_k \circ \alpha_k ( \{ w_j \} ).
\end{equation*}
Since each $w_j$ is included in a $k$-TSS language contained in $L$,
and $\gamma_k \circ \alpha_k ( \{ w_j \} )$ is the smallest $k$-TSS
language that contains $w_j$, we conclude that, for all $j$,
$\gamma_k \circ \alpha_k ( \{ w_j \} ) \subseteq L$, which in turn
implies $K_i \subseteq L$.
Since there are only finitely many $k$-test vectors and finitely many
$k$-TSS languages, the sequence
\begin{equation}
\label{sequence Ks}
K_1 \subseteq K_2 \subseteq K_3 \subseteq \cdots
\end{equation}
converges, that is there exists an $n$ such that, for all $m \geq n$,
$K_m = K_n$. This implies that all $w_j$ are included in $K_n$, that
is $L \subseteq K_n$. In combination with the above observation that all
$K_i$ are contained in $L$, this proves that sequence (\ref{sequence Ks})
converges to $L$.
\end{proof}

The proof of Theorem~\ref{unions of kTSSs learnable in the limit} provides us with
a simple first algorithm to learn unions of $k$-TSS languages: for each example word that we see, we compute the $k$-test vector and then we take the union of the languages denoted by all those $k$-test vectors. The problem with this algorithm is that potentially we end up with a huge number of different $k$-test vectors. Thus we would like to cluster as many $k$-test vectors in the union as we can, without changing the overall language.
Before we can introduce our clustering algorithm, we first need to define a metric
on $k$-test vectors.

\begin{definition}
The \emph{cardinality} of a $k$-test vector $Z = \langle I, F, T, C\rangle$ is defined by:
\[
|Z| = |I| + |F| + |T| + |C \cap \Sigma^{<k-1}|.
\] 
\end{definition}

\begin{definition}
\label{def-distance}
The function $d\colon \mathcal{T}_k \times \mathcal{T}_k \mapsto {\rm I\!R}^+$,
which defines the \emph{distance} between a pair of $k$-test vectors, is 
given by: $d(Z, Z') = | Z \bigtriangleup Z' |$.
\end{definition}

Intuitively, the distance between two $k$-test vectors is the number of prefixes, suffixes, substrings and short words that must be added/removed to transform one
$k$-test vector into the other.
For examples, see Fig. \ref{fig:distance-matrix}.

\begin{lemma}
\label{lemma}
Let $Z$, $Z'$ and $Z''$ be $k$ test vectors. Then
\begin{eqnarray}
\label{triangle1}
| Z \bigtriangleup Z' | & \leq & |Z| + |Z'|\\
\label{triangle2}
Z \bigtriangleup \bot & = & Z\\
\label{triangle3}
Z \bigtriangleup Z'  =  \bot & \Leftrightarrow & Z = Z'\\
\label{trianglecommutes}
Z \bigtriangleup Z' & = & Z' \bigtriangleup Z\\
\label{triangle4}
Z \bigtriangleup (Z' \bigtriangleup Z'') & = & (Z \bigtriangleup Z') \bigtriangleup Z''\\
\label{triangle5}
Z \bigtriangleup Z' & = & (Z \bigtriangleup Z'') \bigtriangleup (Z'' \bigtriangleup Z')
\end{eqnarray}
\end{lemma}
\begin{proof}
Let $Z = \langle I, F, T, C \rangle$, $Z' = \langle I', F', T', C' \rangle$
and $Z'' = \langle I'', F'', T'', C'' \rangle$.
Then equality (\ref{triangle1}) can be derived as follows:
\begin{eqnarray*}
| Z \bigtriangleup Z' | & = & | I \bigtriangleup I' | + | F \bigtriangleup F' |
+ | T \bigtriangleup T' |\\
& & +
| (C \bigtriangleup C' \bigtriangleup (I' \cap F) \bigtriangleup (I \cap F')) \cap \Sigma^{<k-1} |\\
& = & | I \bigtriangleup I' | + | F \bigtriangleup F' |
+ | T \bigtriangleup T' |\\
& &  +
| (C\cap \Sigma^{<k-1}) \bigtriangleup (C'\cap \Sigma^{<k-1}) \bigtriangleup (I' \cap F \cap \Sigma^{<k-1}) \bigtriangleup (I \cap F'\cap \Sigma^{<k-1}) |\\
& = & | I \bigtriangleup I' | + | F \bigtriangleup F' |
+ | T \bigtriangleup T' |\\
& &  +
| (C\cap \Sigma^{<k-1}) \bigtriangleup (C'\cap \Sigma^{<k-1}) \bigtriangleup \emptyset \bigtriangleup \emptyset |\\
& = & | I \bigtriangleup I' | + | F \bigtriangleup F' |
+ | T \bigtriangleup T' | +
| (C\cap \Sigma^{<k-1}) \bigtriangleup (C'\cap \Sigma^{<k-1}) |\\
& \leq & | I | + | I' | + | F | + | F' |
+ | T | + | T' | +
| (C\cap \Sigma^{<k-1}) | + | (C'\cap \Sigma^{<k-1}) |\\
& = & |Z| + |Z'|
\end{eqnarray*}
Identities (\ref{triangle2}), (\ref{triangle3}), (\ref{trianglecommutes}) and (\ref{triangle4}) follow from the definitions and the corresponding identities for sets:
\begin{eqnarray*}
Z \bigtriangleup \bot & = & \langle I \bigtriangleup \emptyset, F \bigtriangleup \emptyset, T \bigtriangleup \emptyset, C \bigtriangleup \emptyset \bigtriangleup (\emptyset \cap F) \bigtriangleup (I \cap \emptyset)\rangle\\
& = & \langle I, F, T, C \rangle = Z\\
Z \bigtriangleup Z' = \bot 
& \Leftrightarrow & I \bigtriangleup I'  = \emptyset \wedge
F \bigtriangleup F'  = \emptyset \wedge
T \bigtriangleup T'  = \emptyset \wedge
C \bigtriangleup C' \bigtriangleup (I' \cap F) \bigtriangleup (I \cap F')  = \emptyset\\
& \Leftrightarrow & I = I' \wedge F = F' \wedge T = T' \wedge 
C \bigtriangleup C' \bigtriangleup (I' \cap F) \bigtriangleup (I \cap F')  = \emptyset\\
& \Leftrightarrow & I = I' \wedge F = F' \wedge T = T' \wedge 
C \bigtriangleup C' \bigtriangleup (I \cap F) \bigtriangleup (I \cap F)  = \emptyset\\
& \Leftrightarrow & I = I' \wedge F = F' \wedge T = T' \wedge 
C \bigtriangleup C' \bigtriangleup \emptyset  = \emptyset\\
& \Leftrightarrow & I = I' \wedge F = F' \wedge T = T' \wedge 
C \bigtriangleup C'   = \emptyset\\
& \Leftrightarrow & I = I' \wedge F = F' \wedge T = T' \wedge C = C'\\
& \Leftrightarrow & Z = Z'\\
Z \bigtriangleup Z' & = & \langle I \bigtriangleup I', F \bigtriangleup F', T \bigtriangleup T', C \bigtriangleup C' \bigtriangleup (I' \cap F) \bigtriangleup (I \cap F')\rangle\\
& = & \langle I' \bigtriangleup I, F' \bigtriangleup F, T' \bigtriangleup T, C' \bigtriangleup C \bigtriangleup (I \cap F') \bigtriangleup (I' \cap F)\rangle\\
& = & Z' \bigtriangleup Z\\
Z \bigtriangleup (Z' \bigtriangleup Z'') & = & 
\langle I \bigtriangleup (I' \bigtriangleup I''),
F \bigtriangleup (F' \bigtriangleup F''),
T \bigtriangleup (T' \bigtriangleup T''),\\
& & C \bigtriangleup ( C' \bigtriangleup C'' \bigtriangleup (I'' \cap F') \bigtriangleup (I' \cap F'')) \bigtriangleup ((I' \bigtriangleup I'') \cap F) \bigtriangleup (I \cap (F' \bigtriangleup F''))\rangle \\
& = & (Z \bigtriangleup Z') \bigtriangleup Z''\\
& = & 
\langle I \bigtriangleup (I' \bigtriangleup I''),
F \bigtriangleup (F' \bigtriangleup F''),
T \bigtriangleup (T' \bigtriangleup T''),\\
& & C \bigtriangleup  C' \bigtriangleup C'' \bigtriangleup (I'' \cap F') \bigtriangleup (I' \cap F'') \bigtriangleup (I' \cap F) \bigtriangleup (I'' \cap F) \bigtriangleup (I \cap F') \bigtriangleup (I \cap F'')\rangle \\
& = & 
\langle (I \bigtriangleup I') \bigtriangleup I'',
(F \bigtriangleup F') \bigtriangleup F'',
(T \bigtriangleup T') \bigtriangleup T'',\\
& & (C \bigtriangleup  C' \bigtriangleup (I' \cap F) \bigtriangleup (I \cap F')) 
\bigtriangleup C'' \bigtriangleup (I'' \cap (F \bigtriangleup F')) \bigtriangleup ((I \bigtriangleup I') \cap F'') \rangle \\
& = & (Z \bigtriangleup Z') \bigtriangleup Z''
\end{eqnarray*}
Finally, equality (\ref{triangle5}) follows from the preceding equalities (\ref{triangle2}), (\ref{triangle3}) and (\ref{triangle4}):
\begin{eqnarray*}
Z \bigtriangleup Z' & = & (Z \bigtriangleup \bot) \bigtriangleup Z'\\
& = & (Z \bigtriangleup (Z'' \bigtriangleup Z'')) \bigtriangleup Z'\\
& = & ((Z \bigtriangleup Z'') \bigtriangleup Z'') \bigtriangleup Z'\\
& = & (Z \bigtriangleup Z'') \bigtriangleup (Z'' \bigtriangleup Z')
\end{eqnarray*}
\end{proof}

\begin{proposition}
\label{prop-distance-metric}
Distance function $d$ is a metric.
\end{proposition}
\begin{proof}
In order to prove that $d$ is a metric, we need to show that it satisfies the following
four properties, for all $Z, Z', Z'' \in \mathcal{T}_k$,
\begin{enumerate}
    \item $d(Z,Z') \geq 0$ (non-negativity)
    \item $d(Z,Z') = 0$ iff $Z = Z'$ (identity of indiscernibles)
    \item $d(Z,Z') = d(Z',Z)$ (symmetry)
    \item $d(Z,Z'') \leq d(Z, Z') + d(Z', Z'')$ (triangle inequality)
\end{enumerate}
Non-negativity follows immediately from the definition of $d$.
Identity of indiscernibles can be derived using Lemma~\ref{lemma}(\ref{triangle3}):
\begin{eqnarray*}
d(Z,Z') = 0 & \Leftrightarrow & | Z \bigtriangleup Z' | = 0
 \Leftrightarrow  Z \bigtriangleup Z'  = \bot
 \Leftrightarrow  Z = Z'.
\end{eqnarray*}
Symmetry follows since $\bigtriangleup$ commutes (Lemma~\ref{lemma}(\ref{trianglecommutes})):
\[
d(Z,Z') = | Z \bigtriangleup Z' | = | Z' \bigtriangleup Z | = d(Z',Z).
\]
The triangle inequality follows using identities (\ref{triangle1})
and (\ref{triangle5}) from Lemma~\ref{lemma}:
\begin{eqnarray*}
d(Z, Z') & = & | Z \bigtriangleup Z' |\\
& = & |(Z \bigtriangleup Z'') \bigtriangleup (Z'' \bigtriangleup Z')|\\
& \leq & |Z \bigtriangleup Z''| + | Z'' \bigtriangleup Z'|\\
& = & d(Z, Z'') + d(Z'', Z').
\end{eqnarray*}
\end{proof}

%In order to establish that function $d$ is a metric, we first show that a number of identities that hold for the symmetric difference operation on sets, can be lifted to $k$-test vectors.

The next proposition provides a necessary and sufficient condition for the union of the languages 
of two $k$-test vectors to be equal to the language of the union of these $k$-test vectors. 
%For proof, see Appendix \ref{appendix:union prop}.

\begin{proposition}
\label{union prop}
Suppose $Z = \langle I, F, T, C \rangle$ and $Z' =  \langle I', F', T', C' \rangle$
are canonical $k$-test vectors, for some $k$.
Let $\bullet \not\in \Sigma$ be a fresh symbol, and let $G=(V,E)$ be the directed graph
with
\begin{eqnarray*}
V & = & \{ \bullet u \mid u \in I \cup I' \} \cup T \cup T' \cup \{ u \bullet \mid u \in F \cup F'\},\\
E & = & \{ (au, ub) \in V \times V \mid a, b \in \Sigma \cup \{ \bullet \}, u \in \Sigma^{k-1} \}.
\end{eqnarray*}
Suppose each vertex in $V$ is colored either red, blue or white.
Vertices in $T \setminus T'$ are red, vertices in $T'\setminus T$ are blue, and vertices
in $T \cap T'$ are white.
A vertex $\bullet u$ is red if $u \in I \setminus I'$, blue if
$u \in I' \setminus I$, and white if $u \in I \cap I'$.
A vertex $u \bullet$ is red if $u \in F \setminus F'$, blue if
$u \in F' \setminus F$, and white if $u \in F \cap F'$.
Then $\gamma_k (Z \sqcup Z') = \gamma_k(Z) \cup \gamma_k(Z')$ iff
there exists no path in $G$ from a red vertex to a blue vertex, nor from a blue vertex to a red vertex.
\end{proposition}

\begin{proof}
Let $\Pi$ be the set of paths in $G$ from a vertex in $\{ \bullet u \mid u \in I \cup I' \}$ to a vertex in
$\{ u \bullet \mid u \in F \cup F'\}$.
There exists a 1-to-1 correspondence between paths in $\Pi$ and strings in
$\gamma_k (Z \sqcup Z')$ with length at least $k-1$. Because suppose 
\[
w = a_1 \cdots a_m
\]
is a string in $\gamma_k (Z \sqcup Z')$, for some $m \geq k-1$.
Let $a_0 = \bullet$, $a_{m+1} = \bullet$ and let, for $0 \leq j \leq m-k+2$,
\[
w_j = a_j  \cdots a_{j+k-1}
\]
be the substring of $a_0 a_1 \cdots a_m a_{m+1}$ with length $k$ starting at $a_j$.
Using the fact that $w \in \gamma_k (Z \sqcup Z')$, the reader can easily check that all the
string $w_j$ are vertices from $V$. Moreover, for all $0 \leq j < m-k+2$, the pair
$(w_j,w_{j+1})$ is an edge of $E$.
Thus we may conclude that sequence $w_0 ,\ldots, w_{m-k+2}$ is a path in $\Pi$.
Conversely, 
it is trivial to extract from each path in $\Pi$ a corresponding string in 
$\gamma_k (Z \sqcup Z')$ with length at least $k-1$.
Moreover, paths in $\Pi$ that only visit red and white vertices correspond with
words in $\gamma_k(Z)$,
paths in $\Pi$ that only visit blue and white vertices correspond with
words in $\gamma_k(Z')$, and
paths in $\Pi$ that contain both a red and a blue vertex correspond with words
in $\gamma_k (Z \sqcup Z') \setminus (\gamma_k(Z) \cup \gamma_k(Z'))$.

\paragraph{$\Rightarrow$}
Suppose there exists a path in $G$ from a red vertex $v$ to a blue vertex $w$. (The case in which
there exists a path from a blue vertex to a red vertex is symmetric.)
Since $Z$ and $Z'$ are canonical, each red vertex in $V$ occurs on a path in $\Pi$ that only
contains red or white vertices, and
each blue vertex in $V$ occurs on a path in $\Pi$ that only contains blue or white vertices.
This means we can construct a path $\pi \in \Pi$ that contains both the red vertex $v$ and the
blue vertex $w$. By the above observations, path $\pi$ corresponds to a word in
$\gamma_k (Z \sqcup Z')$ that is not in $\gamma_k(Z) \cup \gamma_k(Z')$.

\paragraph{$\Leftarrow$}
Now suppose there exists no path in $G$ from a red vertex to a blue vertex, nor from a blue vertex to a red vertex.
Then any path in $\Pi$ either contains only red and white vertices, or it contains only blue and white vertices.
This implies that any word in $\gamma_k (Z \sqcup Z')$ with length at least $k-1$ is contained either in $\gamma_k(Z)$ or in $\gamma_k(Z')$. By construction all the words in $\gamma_k (Z \sqcup Z')$ with length less than $k-1$ are contained either in $C \subseteq \gamma_k(Z)$ or in $C'\subseteq \gamma_k(Z')$.
Thus we may conclude $\gamma_k (Z \sqcup Z') \subseteq \gamma_k(Z) \cup \gamma_k(Z')$.
Since $\gamma_k$ is monotone (Proposition~\ref{gamma monotone}), the converse inclusion
$\gamma_k (Z \sqcup Z') \supseteq \gamma_k(Z) \cup \gamma_k(Z')$ holds as well.
Thus we may conclude $\gamma_k (Z \sqcup Z') = \gamma_k(Z) \cup \gamma_k(Z')$, as required.
\end{proof}

Suppose alphabet $\Sigma$ contains $n$ elements. Then
the size of graph $G$ is in $O(n \cdot |Z \cup Z'|)$, and we can construct $G$ from 
$Z$ and $Z'$ in time $O((n+k) \cdot |Z \cup Z'|)$. Since the reachability property in Proposition~\ref{union prop} can be decided in a time that is linear in the size of $G$, we obtain an $O((n+k) \cdot|Z \cup Z'|)$-time algorithm for
deciding $\gamma_k (Z \sqcup Z') = \gamma_k(Z) \cup \gamma_k(Z')$.

%=====================================
%			    METHOD
%=====================================

\subsection{Efficient algorithm}

Our algorithm to learn unions of $k$-testable languages is based on hierarchical clustering.
Given a set $\mathcal{S}$ of $n$ words, we compute its related set of $k$-test vectors $S = \{ \alpha_k (\{x\}) \mid x \in \mathcal{S} \}$.
Note that the $k$-test vectors are canonical.
Then, an $n \times n$ distance matrix is computed. To that end, the distance used is the pairwise distance between $k$-test vectors defined in Definition~\ref{def-distance}.
Next, the algorithm finds the closest pair of compatible $k$-test vectors $Z$ and $Z'$, such that $\gamma_k (Z \sqcup Z') = \gamma_k(Z) \cup \gamma_k(Z')$ and computes their union.
The distance between the merged $k$-test vectors and the remaining $k$-test vectors in $S$ is updated.
These two operations are repeated until all initial $k$-test vectors have been merged into one, or that no allowed union of two $k$-test vectors such that $\gamma_k (Z \sqcup Z') = \gamma_k(Z) \cup \gamma_k(Z')$ is possible.
We gather at the end of the process a linkage between $k$-test vectors, which can lead to the computation of a dendrogram.
When the number of $k$-test vectors to learn is known, one can use this expected number of languages to find the threshold that would, given the hierarchical clustering, return the desired unions of $k$-test vectors. 

\begin{algorithm}[h]
\caption{Linkage of $k$-test vectors.}
\label{algo:linkage-ktest-vectors}
 \DontPrintSemicolon
\KwIn{$S = \{ \alpha_k (\{x\}) \mid x \in \mathcal{S} \}$ : sample of $m$ $k$-test vectors, $D$: related distance matrix}
\KwOut{$L$: unsorted dendrogram of $k$-test vectors. }

$chain\gets[\,]$ \;

%$size[x]\gets 1$ for all $x\in S$ \;

\While{$|S|>1$ \textbf{and} $ \exists a,b \in S$ s.t. $\gamma_k (a \sqcup b) = \gamma_k(a) \cup \gamma_k(b)$ }{

    \If{$\texttt{length}(chain)\leq3$}{
        $a \gets \text{(any element of $S$)}$ \;
    
        $chain \gets [a]$ \; 
    
        $b \gets \text{(any element of $S\setminus\{a\}$)}$ \;
    }
    \Else{
    
        $a \gets chain[-4]$\;
        
        $b \gets chain[-3]$\;
    
        Remove $chain[-1]$, $chain[-2]$ and $chain[-3]$ \;
    }
    \Do{$\texttt{length}(chain) \geq 3$ \textbf{and} $a = chain[-3]$}{
    
        $c \gets argmin_{x\neq a}d[x,a]$ with preference for $b$ \;
        
        $a,b \gets c,a$ \;
        
        \If{$\gamma_k (a \sqcup b) = \gamma_k(a) \cup \gamma_k(b)$  \tcp*{check if $\gamma_k(a) \cup \gamma_k(b)$ is a k-TSS language}}{
                Append $a$ to $chain$ \;
        }
        
    }

    Append $(a,b, D[a,b])$ to $L$  \;
    
    Remove $a,b$ from $S$ \;
    
    %$ size[a \sqcup b] \gets size[a] + size[b]$ \;
    %$D[n,x] = D[x,n] \gets \max(D[a,x],D[b,x]) , \forall x \in S$  \tcp*{Update D}
    
    $D[a \sqcup b,x] = D[x,a \sqcup b] \gets d(a \sqcup b, x) , \forall x \in S$  \tcp*{Update D}

    $S \gets S \cup \{a \sqcup b\}$ \; \tcp*{new node is the union of the two k-test vectors}

}

\textbf{return} $L$ \;

\end{algorithm}

An efficient implementation for such hierarchical clustering algorithms is the nearest-neighbor chain algorithm \cite{benzecri1982construction}.
%Our hierarchical clustering of $k$-test vectors is presented in Appendix \ref{appendix:algorithm}.
Our hierarchical clustering of $k$-test vectors is presented in Algorithm \ref{algo:linkage-ktest-vectors}.
We reuse notations from \cite{mullner2011modern}.

\begin{example}
Let $k=3$.
Given the sample of strings $\mathcal{S}$ in Table \ref{fig:3-test-vectors}, compute the associate sample of $3$-test vectors $S = \{Z_1,Z_2,\ldots,Z_{8}\}$.
Then, compute its distance matrix (Table \ref{fig:distance-matrix}) using the metric defined in Definition \ref{def-distance}.
Using Algorithm \ref{algo:linkage-ktest-vectors}, compute the related linkage matrix depicted in Table \ref{fig:linkage-matrix}.
We gather the dendrogram shown in Figure \ref{fig:hierarchical-clustering}, where the 3 remaining $3$-test vectors $Z_1 \sqcup Z_8$, $Z_2 \sqcup Z_5 \sqcup Z_7$ and $Z_3 \sqcup Z_4 \sqcup Z_6$ cannot be merged.
Indeed: 
\begin{itemize}
    \item $\gamma_k( Z_1 \sqcup Z_8 \sqcup Z_2 \sqcup Z_5  \sqcup Z_7) \neq  \gamma_k( Z_1 \sqcup Z_8 ) \cup \gamma_k( Z_2 \sqcup Z_5  \sqcup Z_7) $.
    \item $\gamma_k( Z_1 \sqcup Z_8 \sqcup Z_3 \sqcup Z_4 \sqcup Z_6) \neq  \gamma_k( Z_1 \sqcup Z_8 ) \cup \gamma_k( Z_3 \sqcup Z_4 \sqcup Z_6) $.
    \item $\gamma_k( Z_2 \sqcup Z_5 \sqcup Z_7 \sqcup Z_3 \sqcup Z_4 \sqcup Z_6) \neq  \gamma_k( Z_2 \sqcup Z_5  \sqcup Z_7) \cup \gamma_k(Z_3 \sqcup Z_4 \sqcup Z_6 ) $.
\end{itemize}
With a desired number of $3$-TSS languages to learn of 3, the returned languages are $\gamma_k(Z_1 \sqcup Z_8)$ and $\gamma_k(Z_2 \sqcup Z_5 \sqcup Z_7)$  and $ \gamma_k( Z_3 \sqcup Z_4 \sqcup Z_6 )$.
With a desired number of $3$-TSS languages to learn of 4, the returned languages would be $\gamma_k(Z_1 \sqcup Z_8)$ and $\gamma_k(Z_2 \sqcup Z_5 \sqcup Z_7)$  and $ \gamma_k( Z_3 )$ and $ \gamma_k( Z_4 \sqcup Z_6 )$ instead.

\begin{figure}[h!]
\centering
\captionsetup[subfigure]{position=b}
\subcaptionbox{Dataset and corresponding 3-test vectors \label{fig:3-test-vectors}}{
\begin{tabular}{|l|l|}
\hline
$\mathcal{S}$ & $S$ \\ \hline
baba & $Z_1 = \langle \{ba\},\{ba\},\{bab,aba\},\{\} \rangle $ \\ \hline
abba & $Z_2 = \langle \{ab\},\{ba\},\{abb,bba\},\{\} \rangle $ \\ \hline
abcabc & $Z_3 = \langle\{ab\},\{bc\},\{abc,bca,cab\},\{\} \rangle $ \\ \hline
cbacba & $Z_4 = \langle \{cb\},\{ba\},\{cba,bac,acb\},\{\} \rangle $ \\ \hline
abbbba & $Z_5 = \langle \{ab\},\{ab\},\{abb,bbb,bba\},\{\} \rangle $ \\ \hline
cbacbacba & $Z_6 = \langle \{cb\},\{ba\},\{cba,bac,acb\},\{\} \rangle $ \\ \hline
abbba & $Z_7 = \langle \{ab\},\{ba\},\{abb,bbb,bba\},\{\} \rangle $ \\ \hline
babababc & $Z_8 = \langle \{ba\},\{bc\},\{bab,aba,abc\},\{\} \rangle $ \\ \hline
\end{tabular}

}
\subcaptionbox{Distance matrix\label{fig:distance-matrix}}{
\begin{tabular}{|c|c|c|c|c|c|c|c|c|}
\hline
 & $Z_1$ & $Z_2$ & $Z_3$ & $Z_4$ & $Z_5$ & $Z_6$ & $Z_7$ & $Z_8$ \\ \hline
$Z_1$ & 0 & 6 & 9 & 7 & 7 & 7 & 7 & 3 \\ \hline
$Z_2$ & 6 & 0 & 7 & 7 & 1 & 7 & 1 & 9 \\ \hline
$Z_3$ & 9 & 7 & 0 & 10 & 8 & 10 & 8 & 6 \\ \hline
$Z_4$ & 7 & 7 & 10 & 0 & 8 & 0 & 8 & 10 \\ \hline
$Z_5$ & 7 & 1 & 8 & 8 & 0 & 8 & 0 & 10 \\ \hline
$Z_6$ & 7 & 7 & 10 & 0 & 8 & 0 & 8 & 10 \\ \hline
$Z_7$ & 7 & 1 & 8 & 8 & 0 & 8 & 0 & 10 \\ \hline
$Z_8$ & 3 & 9 & 6 & 10 & 10 & 10 & 10 & 0 \\ \hline
\end{tabular}

}
\subcaptionbox{Linkage matrix\label{fig:linkage-matrix}}{
\hspace*{0.2cm}
\begin{tabular}{|c|c|c|}
\hline
$Z_5$ & $Z_7$ & 0 \\ \hline
$Z_4$ & $Z_6$ & 0 \\ \hline
$Z_2$ & $Z_5 \sqcup Z_7$ & 1 \\ \hline
$Z_1$ & $Z_8$ & 3 \\ \hline
$Z_3$ & $Z_4 \sqcup Z_6$ & 10 \\ \hline
\end{tabular}
\vspace*{1cm}
\hspace*{0.2cm}
}
\subcaptionbox{Corresponding dendrogram\label{fig:hierarchical-clustering}}{
\includegraphics[width=0.65\linewidth]{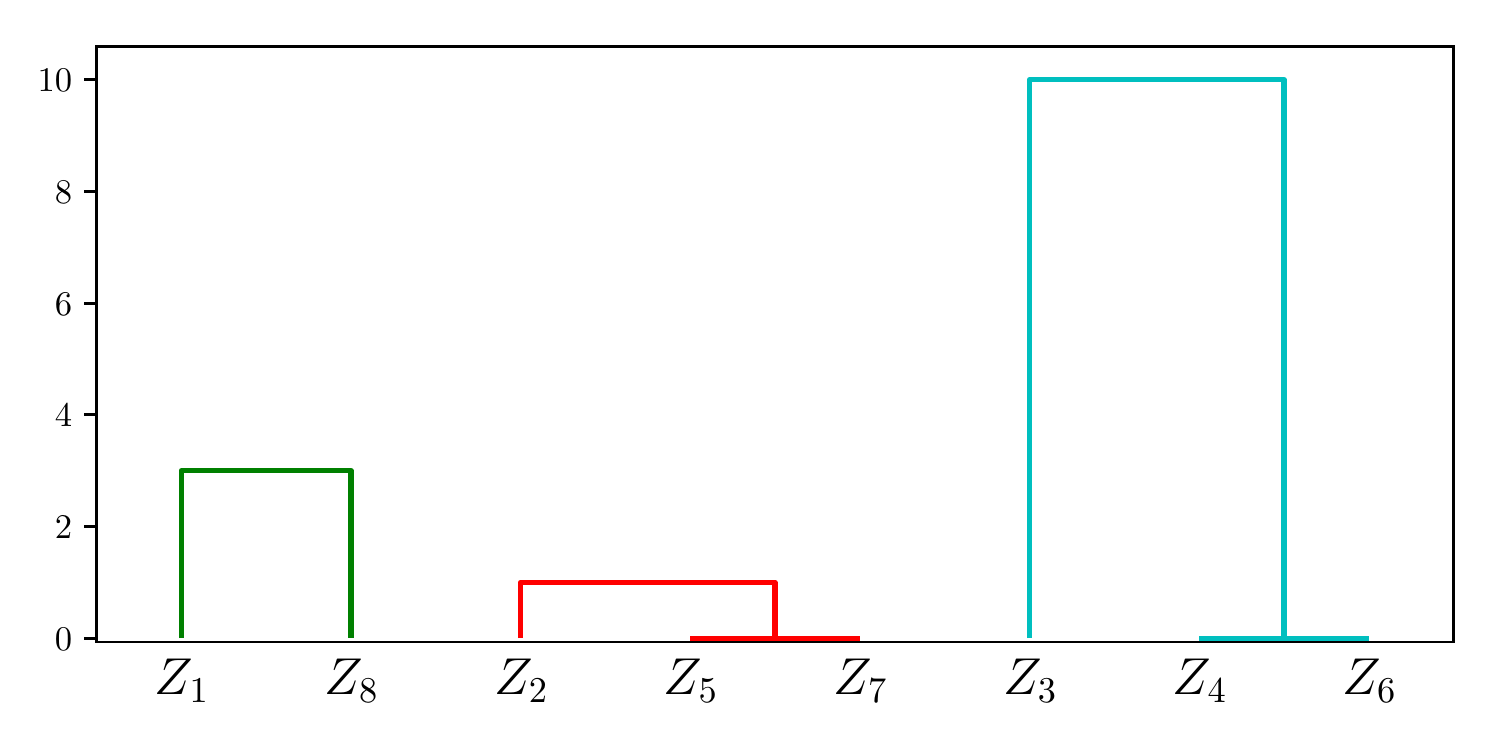}
}
\caption{Learning union of k-test vectors}
\label{fig:example}
\end{figure}

\end{example}

\begin{proposition} 
The time complexity of learning union of k-test vectors using Algorithm \ref{algo:linkage-ktest-vectors} is in $\mathcal{O}(mk \cdot |x| + 2^{n^k} \cdot m^2 + (n+k) \cdot 2^{n^k} \cdot m^2 )$.

\begin{proof}
   Suppose alphabet $\Sigma$ contains $n$ elements and that the sample contains $m$ elements. Let $x$ be the longest word in the sample. The learning of unions of k-test vectors is done in 3 steps:
    \begin{itemize}
        \item Computing the initial $k$-test vectors, which is in $O(mk \cdot |x|)$ with x the longest word in the sample.
        \item Building the distance matrix is in $O(2^{n^k} \cdot m^2)$.
        \item Building the hierarchical clustering (linkage of neighbouring k-test vectors) is in $O((n+k) \cdot 2^{n^k} \cdot m^2)$.
    \end{itemize}
\end{proof}

\end{proposition}

%=====================================
%		Experiments and Results
%=====================================
\section{Case Study}
\label{sec:experiments}

\paragraph{Job dataset}
Our case study has been inspired by an industrial problem related to the domain of cyber-physical systems. 
Recent work \cite{waqas} focused on the impact of design parameters of a flexible manufacturing system on its productivity. 
It appeared in the aforementioned study that the productivity depends on the jobs being rendered. 
To that end, the prior identification of the different job patterns is crucial in enabling engineers to optimize parameters related to the flexible manufacturing system. 
\begin{table}[!h]
\begin{center}
\caption{Sample of identified job patterns.}
\begin{tabular}{lccr}
\toprule
job & pattern   & 3-test vector & type of job \\ 
\midrule
$\ aaaaa$ 						& \multirow{3}{*}{$a^+$}  	& \multirow{3}{*}{$Z=\langle \{aa\},\{aa\},\{aaa\},\{a,aa\} \rangle$}	& \multirow{3}{*}{homogeneous} \\
\cline{1-1}
\rule{0pt}{2.5ex} $aaaaaaaaaa$ & 						 	&					&				\\
\cline{1-1}
\rule{0pt}{2.5ex} $aaaaa\ldots aaa$ &					&						&			\\
\cline{1-4}
\rule{0pt}{2.5ex} $abababab$	& \multirow{2}{*}{$(ab)^+$} & \multirow{2}{*}{$Z=\langle \{ab\},\{ab\},\{aba,bab\},\{ab\} \rangle$}	& \multirow{4}{*}{heterogeneous} 			\\
\cline{1-1}
\rule{0pt}{2.5ex} $abababababab$	&					& 			&						\\
\cline{1-3}
\rule{0pt}{2.5ex} $abcabcabc$	& \multirow{2}{*}{$(abc)^+$}& 		\multirow{2}{*}{$Z=\langle \{ab\},\{bc\},\{abc,bca,cab\},\{\} \rangle$}						&	\\
\cline{1-1}
\rule{0pt}{2.5ex} $abcabcabcabcabc$ & 					&			&						\\
\cline{1-4}
\rule{0pt}{2.5ex} $abcbcbcbca$ & $a(bc)^+a$ 				& $Z=\langle \{ab\},\{ca\},\{abc,bcb,cbc,cba\},\{\} \rangle$  & miscellaneous \\
\bottomrule
\end{tabular}
% \vspace*{-0.1cm}
\label{table:patterns-example}
\end{center}
\end{table}

We consider a dataset containing strings, each representing a job. 
Our job patterns are also represented by 3-testable languages, the 3-test vectors of which are shown in Table \ref{table:patterns-example}.

Our dataset, implementations and complete results are available\footnote{See \href{https://gitlab.science.ru.nl/alinard/learning-union-ktss}{https://gitlab.science.ru.nl/alinard/learning-union-ktss}}.

%=====================================
%			CONCLUSION
%=====================================
\section{Conclusion}
\label{sec:conclusion}
In this paper, we defined a Galois connection characterizing $k$-testable languages.
We also described an efficient algorithm to learn unions of $k$-testable languages that results from this Galois connection.
%In this paper, we described an efficient algorithm to learn unions of $k$-testable languages.
From a practical perspective, we see that obtaining more than one representation is meaningful since a too generalized solution is not necessarily the best.
To avoid unnecessary generalizations, the union of two $k$-testable languages that would not be a $k$-testable language is not allowed.
Note also that depending on the applications, expert knowledge can provide an indication on the number of languages the returned union should contain.
In further work, we would like to extend the learning of unions of languages to regular languages. 
An attempt to learn pairwise disjoint regular languages has been made in \cite{linard2018learning,linard2017learning}.
However, no learnability guarantee has been provided so far.

\bibliography{biblio}

\end{document}